\documentclass[11pt, reqno]{amsart}
\usepackage{amsmath,amssymb, mathrsfs}
\usepackage{graphicx, color}
\usepackage{slashed}
\newcommand{\RN}[1]{%
  \textup{\uppercase\expandafter{\romannumeral#1}}%
}

\newcommand{\R}{\mathbb{R}}
\newtheorem{theorem}{Theorem}
\newtheorem{proposition}[theorem]{Proposition}

\newtheorem{lemma}[theorem]{Lemma}

\theoremstyle{definition}
\newtheorem{definition}[theorem]{Definition}

\usepackage{xpatch}
\usepackage{hyperref}
\usepackage{footnotebackref}

\makeatletter
\xpretocmd{\@adminfootnotes}{\let\@makefntext\BHFN@OldMakefntext}{}{}
\renewcommand\@makefntext[1]{%
  \@ifundefined{@makefnmark}
    {}
    {%
     \renewcommand\@makefnmark{%
       \mbox{%
         \textsuperscript{%
           \normalfont
           \hyperref[\BackrefFootnoteTag]{\@thefnmark}%
         }%
       }\,%
     }%
     \BHFN@OldMakefntext{#1}%
  }%
}
\makeatother

\begin{document}

\title[mass at null infinity]{Total mass and limits of quasi-local \\ mass at
future null infinity}

\author{Mu-Tao Wang}
\address{Mu-Tao Wang\\
Department of Mathematics\\
Columbia University, USA}
\email{mtwang@math.columbia.edu}

\thanks{This material is based upon work supported by the National Science
Foundation under Grant Number DMS 1810856 (Mu-Tao\ Wang).  The author would like to thank Po-Ning, Chen, Jordan Keller, and Ye-Kai Wang for helpful discussions, and in particular, Ye-Kai Wang for going over a preliminary version of this article and suggestions.}

\maketitle
\begin{abstract}
An idealized observer of an astronomical event is situated at future null infinity, where light rays emitted from the source approach. Mathematically, null infinity corresponds to the portion of the spacetime boundary defined by equivalence classes of null geodesics. But what can we observe at future null infinity? In the note, we start by reviewing the description of null infinity in terms of the Bondi-Sachs coordinate system. In particular, we reformulate  
and extend the invariance of Bondi mass and the mass loss formula in terms of a modified mass aspect 2-form. 
At the end, we discuss new results about the limit of quasilocal mass of unit spheres at null infinity in a Bondi-Sachs coordinate system. This is based on joint work with Po-Ning Chen, Jordan Keller, Ye-Kai Wang, and Shing-Tung Yau. 
\end{abstract}
%\maketitle

\section{Introduction}
This note is about mathematical general relativity and we shall discuss some recent applications of the theory of quasilocal mass \cite{Wang-Yau1, Wang-Yau2} to the study of gravitational radiation. We consider an isolated gravitating system that corresponds to an asymptotically flat spacetime, where gravity is weak at infinity.  The future null infinity, denoted as $\mathscr{I}^+$, consists of equivalent classes of future null geodesics \footnote{This is similar to the definition of ideal boundary of a Cartan-Hadamard manifold as equivalent classes of geodesic rays. In physics literature, null infinity arises as a conformal compactification of the physical spacetime introduced by Penrose, see \cite{Penrose3, Penrose4, Geroch}}.

\begin{center}
{\includegraphics[height=5cm, angle=0]{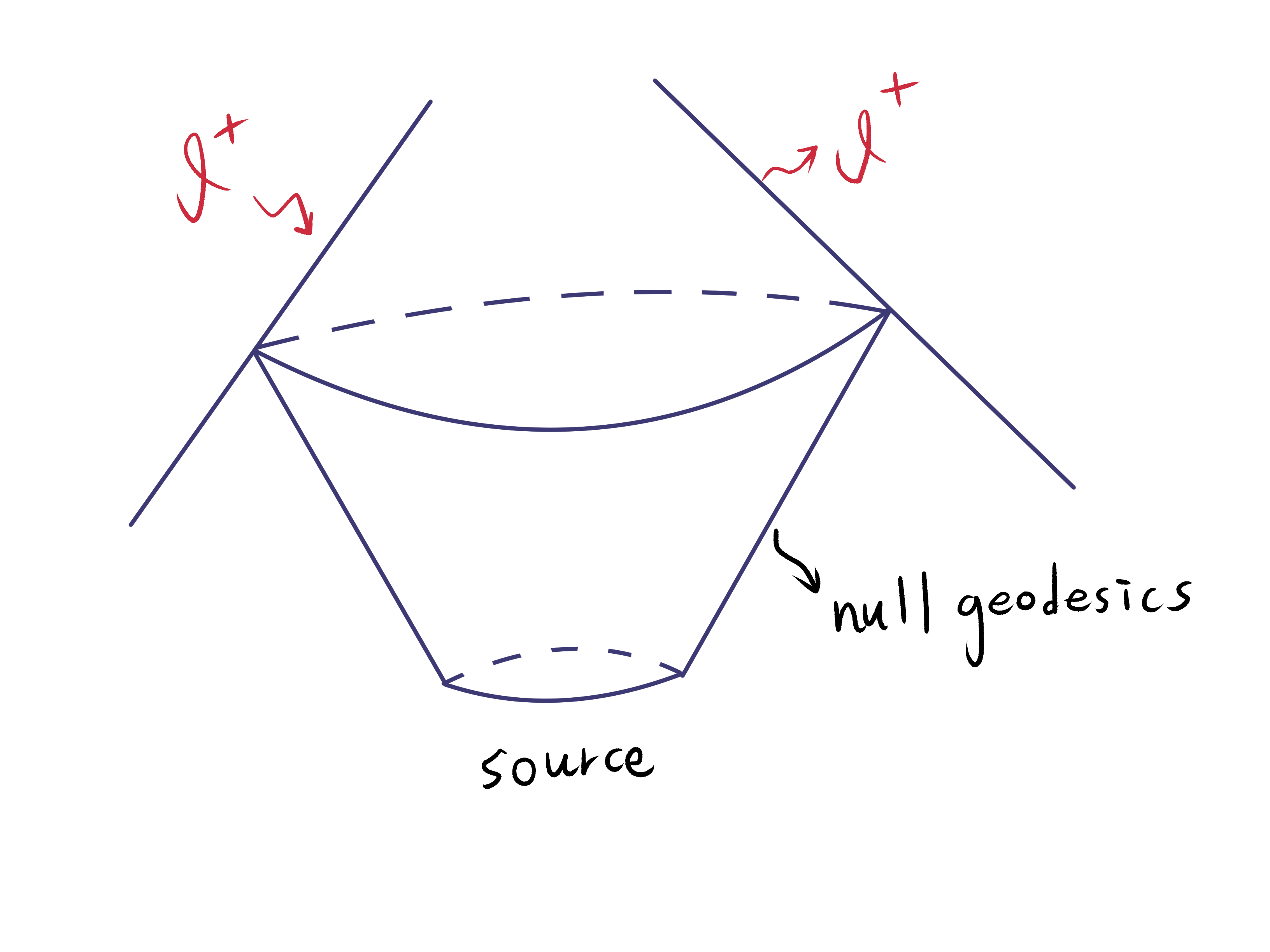}}
\end{center}
 The mathematical (nonlinear) theory of gravitational radiation started around 1960's with the pioneering works of Bondi, Trautman, Sachs, van der Burg, Metzner, Penrose, Newman, Geroch etc \cite{Bondi, Trautman1, Trautman2, BVM, Penrose3, VDB, Sachs, NP1, NP2, Geroch}. The spacetime near $\mathscr{I}^+$ is described in terms of the Bondi-Sachs coordinates $(u, r, \theta, \phi)$ which are chosen in the following way.  Level sets of $u$ are null hypersurfaces generated by null geodesics, $\theta$ and $\phi$ are extended by constancy along the integral curves of the gradient vector field of $u$, and 
 $r$ corresponds to the ``area distance".\footnote{$s=2\ln r$ corresponds to the parameter of the inverse mean curvature (lapse) flow, which was an important tool in the study of the null Penrose inequality, see \cite{Sauter}, as well as the Riemannian Penrose inequality, see \cite{HI}.}
 
 More precisely, one starts with $(u, \theta, \phi)$ coordinates that satisfy the above conditions. For any other coordinate variable $\bar{r}$, in terms of the coordinates $(u, \bar{r}, \theta, \phi)$, the above conditions imply the metric is of the form 
 \begin{equation}-\bar{U}\bar{V} du^2-2\bar{U} dud\bar{r}+\bar{g}_{ab}(dx^a+\bar{W}^a du)(dx^b+\bar{W}^b du), \end{equation}
where $a, b=2, 3$ and $x^2=\theta, x^3=\phi $. Let $\sigma_{ab}$ be a standard round metric on the unit 2-sphere $(S^2, x^a)$. We define a new coordinate variable $r$ by
\[ r=(\det \bar{g}_{ab}/\det \sigma_{ab})^{1/4}\] and denote $ h_{ab}=r^{-2} \bar{g}_{ab}$.

\begin{center}
{\includegraphics[height=5cm]{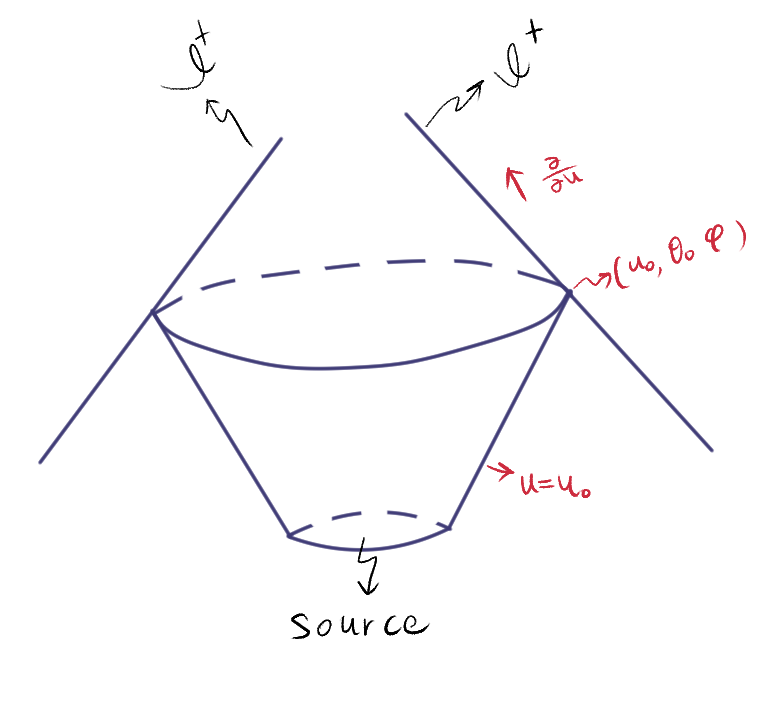}}
\end{center}
In terms of the Bondi-Sachs coordinate system $(u, r,  x^2, x^3)$, the spacetime metric takes the form
\begin{equation}\label{spacetime_metric}g_{\alpha\beta}dx^\alpha dx^\beta= -UV du^2-2U dudr+r^2 h_{ab}(dx^a+W^a du)(dx^b+W^b du).\end{equation} The index conventions here are $\alpha, \beta=0,1, 2, 3$, $a, b=2, 3$, and $u=x^0, r=x^1$. See \cite{CMS, MW} for more details of the construction of the coordinate system. 
The metric coefficients $U, V, h_{ab}, W^a$  of \eqref{spacetime_metric} depend on $u, r, \theta, \phi$, but the assumption on $r$ implies that $\det h_{ab}$ is independent of $u$ and $r$. These gauge conditions thus reduce the number of metric coefficients of a Bondi-Sachs coordinate system to six (there are only two independent components in $h_{ab}$). On the other hand, the boundary conditions $U\rightarrow 1$, $V\rightarrow 1$, $W^a\rightarrow 0$, $h_{ab}\rightarrow \sigma_{ab}$ are imposed as $r\rightarrow \infty$ (such boundary conditions may not be satisfied in a radiative spacetime).   The special gauge choice of the Bondi-Sachs coordinates implies a hierarchy among the vacuum Einstein equations, see \cite{MW, HPS}.

 Assuming the outgoing radiation condition \cite{Sachs, VDB, MW, Winicour, VK}, the boundary condition and the vacuum Einstein equation imply that as $r\rightarrow \infty$, all metric coefficients can be expanded in inverse integral powers of $r$.\footnote{The outgoing radiation condition assumes the traceless part of the $r^{-2}$ term in the expansion of $h_{ab}$ is zero. The presence of this traceless term will lead to a logarithmic term in the expansions of $W^a$ and $V$. Spacetimes with metrics which admit an expansion in terms of $r^{-j}\log^i r$ are called ``polyhomogeneous" and are studied in \cite{CMS}. They do not obey the outgoing radiation condition or the peeling theorem \cite{VK}, but they do appear as perturbations of the Minkowski spacetime by the work of Christodoulou-Klainerman \cite{CK}.} In particular, 
\[\begin{split} U&=1+O(r^{-2}),\\
V&=1-\frac{2m}{r}+O(r^{-2}),\\
W^a&=O(r^{-2}),\\
h_{ab}&={\sigma}_{ab}+\frac{C_{ab}}{r}+O(r^{-2})\end{split},\] where  $m=m(u, x^a)$ is the mass aspect and $C_{ab}=C_{ab}(u, x^a)$ is the shear tensor of this Bondi-Sachs coordinate system.
The Bondi-Sachs  energy-momentum 4-vector associated with a $u$-slice is then \[e(u)=\frac{1}{4\pi}\int_{S^2_\infty} m(u, x^a) dv_\sigma,\,\, p_i(u)=\frac{1}{4\pi}\int_{S^2_\infty} m(u, x^a) Y_i dv_\sigma, i=1,2, 3\] where $\{Y_i=Y_i(x^a), i=1, 2, 3\}$ is an orthonormal basis of the $(-2)$ eigenspace of $\Delta=\Delta_\sigma$ (these are the usual $\ell=1$ spherical harmonics) and $dv_\sigma$ is the area form of the metric $\sigma$. The positivity of the Bondi mass \footnote{This was sometimes called the Bondi-Trautman (or Trautman-Bondi) mass. However, comparison with Trautman's work \cite{Trautman1, Trautman2} and Bondi's work \cite{Bondi, BVM} shows that Trautman worked with a
linear approximation, while Bondi used the full nonlinear theory of General Relativity. The setting of this article is fully nonlinear.} \begin{equation}\label{pmt} \sqrt{e^2-\sum_i p_i^2}\end{equation} was proved by Schoen-Yau \cite{SY} and Horowitz-Perry \cite{HP} under the dominant energy condition and a global assumption on horizon, see also \cite{CJS} and \cite{HYZ}. In Section 4 there is a discussion about the positivity of the Bondi mass and this global assumption. 
The supplementary equations imply the following equation satisfied by the mass aspect along $\mathscr{I}^+$:
\begin{equation}\label{du_mass_aspect}\partial_u m=\frac{1}{4}\nabla^a\nabla^b (\partial_u C_{ab})-\frac{1}{8}|\partial_u C|_\sigma^2,\end{equation} where $|\partial_u C|_\sigma^2=\sigma^{ac}\sigma^{bd} \partial_u C_{ab}\partial_u C_{cd}$. 
Integrating over $S^2_\infty$ with the metric $\sigma$ yields the well-known Bondi mass loss formula:
\begin{equation}\label{energy_loss1}\frac{d}{du} e(u)=-\frac{1}{32\pi}\int_{S^2_\infty} |\partial_u C|_\sigma^2 dv_\sigma \leq 0.\end{equation} In particular,
\begin{equation} \label{energy_loss2} e(u_1)\leq e(u_0) \text{ if } u_1\geq u_0.\end{equation}

\begin{center}
{\includegraphics[height=5cm]{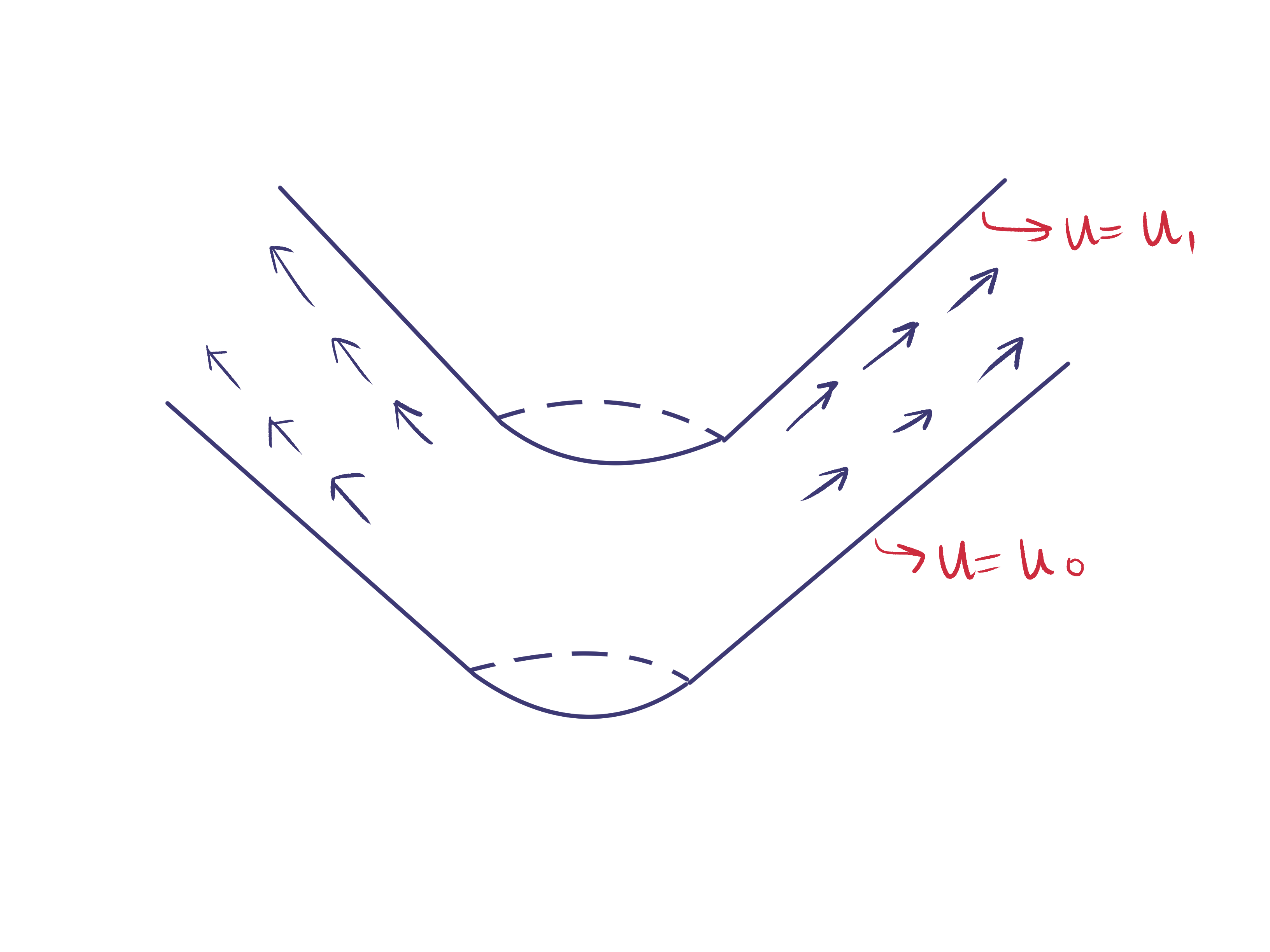}}
\end{center}

This formula indeed corresponds to energy loss, see \cite{HYZ} for a monotonicity formula for the quantity $e-\sqrt{\sum_i p_i^2}$.

\section{Invariance under the BMS group}
Rescaling the spacetime metric \eqref{spacetime_metric} by $r^{-2}$ as $r\rightarrow \infty$, the limit of $r^{-2}g_{\alpha\beta} dx^\alpha dx^\beta$ approaches  $ \sigma_{ab} dx^a dx^b$, or the null metric on $\mathscr{I}^+$.\footnote{This is a special case of conformal compactification. In general, the metric on the unphysical spacetime is of the form $\Omega^2 g_{\alpha\beta} dx^\alpha dx^\beta$ and $\Omega=0$ corresponds to $\mathscr{I}^+$, see \cite{Penrose3, Penrose4, Geroch}.}  Therefore, $\mathscr{I}^+$ can be view as a null three-manifold:  
\[ \mathscr{I}^+= I\times (S^2, \sigma_{ab})\] with $u\in I$, $x^a\in S^2$. 

Each spacetime Bondi-Sachs coordinate system $(u, r, x^a)$ induces such a limiting coordinate system $(u, x^a)$ on $\mathscr{I}^+$, together with the mass aspect $m(u, x^a)$ and the shear $C_{ab}(u, x^a)$. Such a Bondi-Sachs coordinate system is by no means unique and the BMS group, which corresponds to the diffeomorphism group that preserves the gauge and boundary conditions, acts on the set of Bondi-Sachs coordinate systems. 

 A BMS group element induces a diffeomorphism $\mathfrak{g}$ on $\mathscr{I}^+$ that is of the following form:
\begin{equation}\label{bms1} \mathfrak{g}: (u, x^a)\mapsto (\bar{u}, \bar{x}^A), a=2, 3, A=2, 3\end{equation}
such that \begin{equation}\label{bms2} \begin{cases} \bar{x}^A&=g^A(x^a)\\ \bar{u}&=K(x^a)(u+f(x^a))\end{cases}\end{equation}
where $g: (S^2, \sigma)\rightarrow (S^2, \bar{\sigma})$ is a conformal isometry, i.e. $g^*\bar{\sigma}=K^2 \sigma$ where $K=(\alpha_0+\alpha_i Y_i)^{-1}$ and $(\alpha_0, \alpha_i)$ is a future timelike unit vector. 

Here is how the Poincar\`e group sits in the BMS group:

(1)  $f(x^a)$ is any smooth function on $S^2$ that is called a ``supertranslation".  $f(x^a)=\sum a_i Y_i$ corresponds to an actual translation in the Poincar\`e group.

(2)  $K=(\alpha_0+\sum \alpha_i Y_i)^{-1}$ corresponds to boosts in $O(3, 1)$. 

(3) Choices of $Y_i, i=1, 2, 3$ correspond to $O(3)\subset O(3,1)$.

%\begin{frame}
The invariance/equivariance of the Bondi-Sachs energy-momentum is best described in terms of the modified mass aspect which is defined as:
\begin{equation}\label{mma}\widehat{m}=m-\frac{1}{4}\nabla^a\nabla^b C_{ab}, \end{equation} where $\nabla$ is the covariant derivative with respect to the metric $\sigma$.

Let $\widehat{m}$ and $\widehat{\bar{m}}$ be the modified mass aspects of the limiting Bond-Sachs coordinate systems $(u, x^a)$ and $(\bar{u}, \bar{x}^A)$ on $\mathscr{I}^+$, respectively. Suppose 
$(u, x^a)$ and $(\bar{u}, \bar{x}^A)$ are related by a BMS element $(K, f)$ as in \eqref{bms2}, then
 $\widehat{m}$ and $\widehat{\bar{m}}$ are related by
\begin{equation}\label{mass_aspect_transform} K^{-3}(\widehat{m}-\frac{1}{4} \Delta (\Delta+2)f)=\mathfrak{g}^* \widehat{\bar{m}}, \end{equation} see \cite[Section 6.9]{CJK} for the special case of this formula when $K=1$. In the next  subsection, we show that this formula implies the invariance of the Bondi-Sachs energy-momentum.
 In addition, equation \eqref{du_mass_aspect} becomes 
\begin{equation}\label{du_modified_mass}\partial_u\widehat{m}=-\frac{1}{8}|\partial_u C|_{\sigma}^2,\end{equation} and the modified mass aspect is pointwise non-increasing (mass loss formula). Note that  $m$ and $\widehat{m}$ define the same energy-momentum.

\subsection{BMS invariance}

In this subsection, let $\sigma=\sigma_{ab}$ be a round metric on a unit sphere $S^2$ and let $Y_i$ be an orthonormal basis of the $(-2)$ eigenspace of $\Delta_\sigma$ in the sense that 
\begin{equation}\label{orthonormal}\nabla Y_i \cdot \nabla Y_j=\delta_{ij}-Y_i Y_j, i, j=1, 2, 3,\end{equation} where $\nabla=\nabla_\sigma$ is the gradient operator of $\sigma$.  

Suppose $(\alpha_0, \alpha_1, \alpha_2, \alpha_3)$ is a future timelike unit 4-vector, i.e. $\alpha_0>0$ and $\alpha_0^2-\sum_i \alpha_i^2=1$ and denote $K=(\alpha_0+\sum_j \alpha_j Y_j)^{-1}$. It is well-known that the conformal metric $\bar{\sigma}=K^2 \sigma$ is of constant Gauss curvature $1$.

\begin{lemma} Let $\eta_{\alpha\beta}$ be the Minkowski metric. 
  Let $A^\alpha_\beta$ be an element of $O(3,1)$ such that $\eta_{\alpha\beta}=\eta_{\gamma\sigma} A^\gamma_\alpha A^\sigma_\beta$ (which also implies $\eta^{\gamma\sigma}=\eta^{\alpha\beta} A^\gamma_\alpha A^\sigma_\beta$) and $A^0_0=\alpha_0, A_0^k=\alpha_k, k=1, 2, 3$.  Then \[\bar{Y}_i=K(A^0_i+A^k_i Y_k), i=1, 2, 3\] form an orthonormal basis of the $-2$ eigenspace of the Laplace operator $\bar{\Delta}$ of $\bar{\sigma}$ in the sense that 
\[\bar{\nabla} \bar{Y}_i \cdot \bar{\nabla}  \bar{Y}_j    =\delta_{ij}- \bar{Y}_i \bar{Y}_j, i,j=1, 2, 3\]

\end{lemma}

\begin{proof}
This is a direct calculation using the definition $\bar{\sigma}=K^2 \sigma$, the formula $K^{-1}=\alpha_0+\sum_j \alpha_j Y_j$, and the relation \eqref{orthonormal}.\end{proof}
The following proposition shows how the energy-momentum transforms under a boost element of  the BMS group.

\begin{proposition}

Let $\mu$ be a function on $S^2$ and denote \[ {e}=\int \mu dv_{{\sigma}}, \,\, {p}_i=\int \mu Y_i dv_{{\sigma}},\]  where $dv_\sigma$ is the volume form of $\sigma$. 
Suppose $A^\alpha_\beta\in O(3,1)$ satisfies $A_0^0=\alpha_0, A_0^k=\alpha_k$ and let $\bar{Y}_i= (A_i^0+A_i^k Y_k)K, i=1, 2, 3$. Denote 
\[ \bar{e}=\int (K^{-3}\mu) dv_{\bar{\sigma}},\,\, \bar{p}_i=\int ( K^{-3} \mu \bar{Y}_i) dv_{\bar{\sigma}},\] where $dv_{\bar{\sigma}}$ is the volume form of $\bar{\sigma}$.  Then
\[\bar{e}=A_0^0e+A_0^k p_k,    \,    \bar{p}_i= A^0_i e+A_i^k p_k.\] In particular, $\bar{e}^2-\sum \bar{p}_i^2=e^2-\sum p_i^2$. 
\end{proposition}

\begin{proof}
From the last lemma, we know that $\bar{Y}_i= K(A_i^0+A_i^k Y_k), i=1, 2, 3$ forms an orthonormal basis for the $(-2)$ eigenspace of $\Delta_{\bar{\sigma}}$. Since $dv_{\bar{\sigma}}=K^2 dv_{{\sigma}}$, we compute 

\[ \bar{e}=\int (K^{-3}\mu ) dv_{\bar{\sigma}}=\int (A_0^0+A_0^k Y_k) \mu dv_\sigma=A_0^0 \int \mu dv_\sigma+A_0^k \int Y_k \mu dv_\sigma\] and 

\[\bar{p}_i=\int (K^{-3} \mu)  \bar{Y}_i dv_{\bar{\sigma}}=\int (A_i^0+A_i^k Y_k) \mu dv_\sigma=A_i^0 \int \mu dv_\sigma+A_i^k \int Y_k \mu dv_\sigma.\]

Write $p_0=e, \bar{p}_0=\bar{e}$, the formula becomes $\bar{p}_\alpha=A_\alpha^\beta p_\beta, \alpha, \beta=0, 1, 2, 3$. Since we require $\eta_{\alpha\beta}=\eta_{\gamma\sigma} A^\gamma_\alpha A^\sigma_\beta$, the last formula is exactly how a co-vector transforms (a vector $v^\alpha$ transforms by $\bar{v}^\beta=A_\alpha^\beta v^\alpha$). 
\end{proof}

\section{A modified mass aspect 2-form}
\begin{definition}
The modified mass aspect 2-form $\mathfrak{m}$ of a limiting Bondi-Sachs coordinate system $(u, x^a, \sigma)$ of $\mathscr{I}^+$ is defined to be
\[\mathfrak{m}=\widehat{m} dv_\sigma\] where $\widehat{m}$ is the modified mass aspect $\widehat{m}$ defined in \eqref{mma} and $dv_\sigma=\sqrt{\det \sigma} dx^2\wedge dx^3$ is the volume form of the Riemannian metric $\sigma$. 
\end{definition}

In terms of the modified mass aspect 2-form, equation \eqref{mass_aspect_transform} becomes 
\begin{equation}\label{transformation_form}K^{-1}(\mathfrak{m}-\frac{1}{4} \Delta(\Delta+2) f dv_\sigma)=\mathfrak{g}^* \bar{\mathfrak{m}},\end{equation} where $\bar{\mathfrak{m}}=\widehat{\bar{m}} dv_{\bar{\sigma}}$ is the modified mass aspect 2-form of the limiting Bondi-Sachs coordinate system $(\bar{u}, \bar{x}^A, \bar{\sigma})$.

We can then integrate both sides of \eqref{transformation_form} on any section $\Sigma$ of $\mathscr{I}^+$ that is of the form $u=h(x^a)$ where $h$ is any continuous function.

\begin{proposition} For any section $\Sigma$ of $\mathscr{I}^+$, suppose $\mathfrak{m}$ and $\bar{\mathfrak{m}}$ are the modified mass aspect 2-forms of two Bondi-Sachs coordinate systems which are related by a BMS group element that is a pure supertranslation, then the energy integrals are the same \[\int_\Sigma \mathfrak{m}=\int_\Sigma \bar{\mathfrak{m}}.\]
In general, the energy-momentum are related by the boost associated with $K$. 
\end{proposition}

We note that in this formulation, $\Sigma$ does not need to be the level set of any Bondi-Sachs coordinate $u$ on $\mathscr{I}^+$.

The energy loss formula \eqref{energy_loss2} can also be extended by means of the modified mass aspect two-form. 

\begin{definition}For any two sections $\Sigma_1$ and $\Sigma_2$ on $\mathscr{I}^+$, $\Sigma_1$ is said to be in the retarded future of $\Sigma_2$ if there exists a limiting Bondi-Sachs coordinate system $(u, x^a)$ such that $\Sigma_1$ and $\Sigma_2$ are given by $u=h_1 (x^a)$ and $u=h_2 (x^a)$ respectively, and that $h_1 (x^a)\geq h_2 (x^a)$ for each $x^a\in S^2$. 
\end{definition}
One easily check that this notion is independent of the choice of the limiting Bondi-Sachs coordinate system because $K>0$ and \eqref{bms2}. 

\begin{theorem} For any two sections $\Sigma_1$ and $\Sigma_2$ on $\mathscr{I}^+$ such that $\Sigma_1$ is in the retarded future of $\Sigma_2$, we have 
\[\int_{\Sigma_1} \mathfrak{m}\leq \int_{\Sigma_2} \mathfrak{m}.\]
 
\end{theorem}

\begin{proof}

Equation \eqref{du_modified_mass} implies that the mass aspect 2-form $\mathfrak{m}$, as a 2-form on the three-manifold $\mathscr{I}^+$, verifies
\begin{equation}\label{exterior_d} d\mathfrak{m}=-\frac{1}{8} |\partial_u C |_\sigma^2 du\wedge dv_\sigma , \end{equation} where $d$ is the exterior derivative operator on $\mathscr{I}^+$ as a differentiable manifold.

 Integrating both sides of \eqref{exterior_d} over the region between $\Sigma_1$ and $\Sigma_2$  and applying Stokes's theorem yield
\begin{equation}\label{general_loss}\int_{\Sigma_1}\mathfrak{m}-\int_{\Sigma_2} \mathfrak{m}=-\frac{1}{8}\int_{u=h_2(x^a)}^{u=h_1(x^a)} \int_{(S^2, \sigma)} |\partial_u C |_\sigma^2  du\wedge dv_\sigma,\end{equation} which is non-positive by the retarded future condition. 

\end{proof}

This theorem should be considered as an extension of the classical Bondi mass loss formula \eqref{energy_loss2} which only applies to the case when $\Sigma_1$ and $\Sigma_2$  are both smooth and $u$ level sets of a fixed Bondi-Sachs coordinate system. 
\begin{center}
{\includegraphics[height=5cm, angle=0]{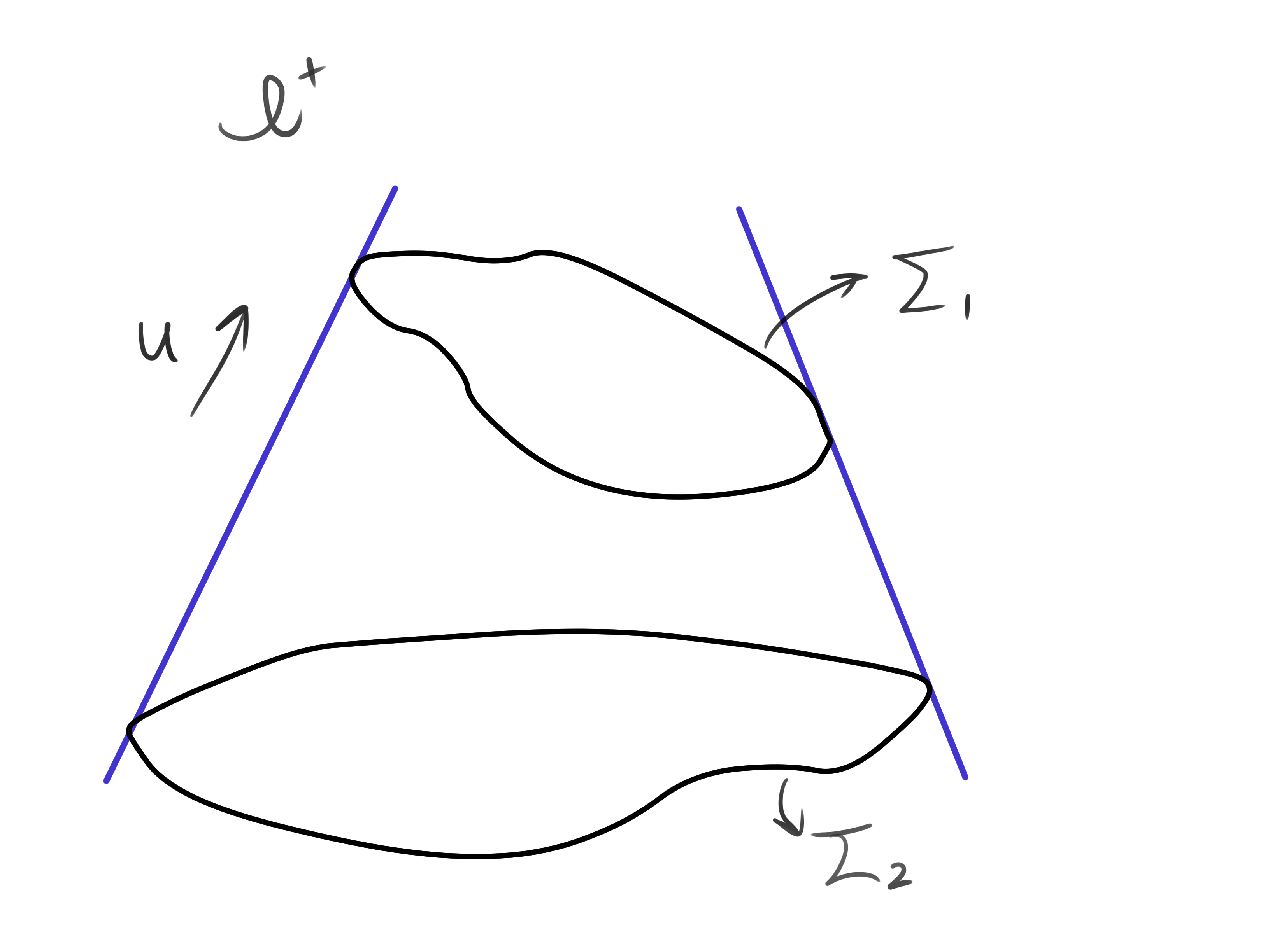}}
\end{center}

\section{Wang-Yau quasilocal mass}
 The quasilocal mass is attached to a 2-dimensional closed surface $\Sigma$ which bounds a spacelike region in spacetime. $\Sigma$ is assumed to be a topological 2-sphere, but with different intrinsic geometry and extrinsic geometry, we expect to read off the effect of gravitation in the spacetime vicinity of the surface. Suppose the surface is spacelike, i.e. the induced metric $\sigma$ is Riemannian.  An essential part of the extrinsic geometry is measured by the mean curvature vector field $\bf{H}$ of $\Sigma$. $\bf{H}$ is a normal vector field of the surface such that the null expansion along any null normal direction $\ell$ is given by the pairing of
$\bf{H}$ and $\ell$.

In \cite{Wang-Yau1}, Wang-Yau proposed the following definition of quasilocal mass.  To evaluate the quasilocal mass of a 2-surface $\Sigma$ with the physical data $(\sigma, \bf{H})$, one first solves the optimal isometric embedding equation, see \eqref{oiee} below, which gives an embedding of $\Sigma$ into the Minkowski spacetime with the image
surface $\Sigma_0$ that has the same 
induced metric as $\Sigma$, i.e. $\sigma$. One then compares the extrinsic geometries of $\Sigma$ and $\Sigma_0$ and evaluates the quasilocal mass from $\sigma, \bf{H}$ and $\bf{H_0}$.

Assuming the mean curvature vector ${\bf H}$ is spacelike, the physical surface $\Sigma$ with physical data $(\sigma, \bf{H})$ gives $(\sigma, |\bf{H}|, \alpha_{\bf {H}})$ where $|{\bf H}|>0$ is the Lorentz norm of $\bf{H}$ and $\alpha_{\bf H}$ is the connection one-form determined by $\bf{H}$. Given an isometric embedding $X:\Sigma\rightarrow \R^{3,1}$ of $\sigma$. Let $\Sigma_0$ be the image $X(\Sigma)$ and $(\sigma, |\bf{H}_0|, \alpha_{\bf {H}_0})$ be the corresponding data of $\Sigma_0$.

Let $T$ be a future timelike unit Killing field of $\R^{3,1}$ and define $\tau=-\langle X, T\rangle$. Define a function $\rho$ and a 1-form $j_a$ on $\Sigma$:
  \[ \begin{split}\rho &= \frac{\sqrt{|{\bf H}_0|^2 +\frac{(\Delta \tau)^2}{1+ |\nabla \tau|^2}} - \sqrt{|{\bf H}|^2 +\frac{(\Delta \tau)^2}{1+ |\nabla \tau|^2}} }{ \sqrt{1+ |\nabla \tau|^2}}\\
 j_a&=\rho {\nabla_a \tau }- \nabla_a \left( \sinh^{-1} (\frac{\rho\Delta \tau }{|{\bf H}_0||{\bf H}|})\right)-(\alpha_{{\bf H}_0})_a + (\alpha_{{\bf H}})_a, \end{split}\] where $\nabla_a$ is the covariant derivative with respect to the metric $\sigma$, $|\nabla \tau|^2=\nabla^a \tau\nabla_a \tau$ and $\Delta \tau=\nabla^a\nabla_a \tau$. 
$\rho$ is the quasilocal mass density and $j_a$ is the quasilocal momentum density. A full set of quasilocal conserved quantities was defined in \cite{Chen-Wang-Yau3, Chen-Wang-Yau4} using $\rho$ and $j_a$. 
 
The optimal isometric embedding equation for $(X, T)$ is 
\begin{equation}\label{oiee} \begin{cases}
\langle dX, dX\rangle&=\sigma\\
\nabla^a j_a&=0.
\end{cases}\end{equation}
The first equation is the isometric embedding equation into the Minkowski spacetime and the second one is the Euler-Lagrange equation of the quasilocal energy \cite{Wang-Yau1,Wang-Yau2} in the space of isometric embeddings.
The quasi-local mass is defined to be \[E(\Sigma, X, T)=\frac{1}{8\pi}\int_\Sigma \rho.\]

 $\Sigma_0$ is essentially the ``unique" surface in the Minkowski spacetime that best matches the physical surface $\Sigma$. 
If the original surface $\Sigma$ happens to be a surface in the Minkowski spacetime, the above procedure identifies $\Sigma_0=\Sigma$ up to a global isometry. In \cite{Chen-Wang-Yau2}, we prove such a minimizing and uniqueness property for a solution of the optimal isometric embedding equation. 

A prototype form of the quasilocal mass (see Brown-York \cite{BY2}, Liu-Yau \cite{Liu-Yau1}, Booth-Mann \cite{Booth-Mann}, Kijowski \cite{KI}, etc) is 
\[\frac{1}{8\pi}\int_\Sigma |\bf H_0|-\frac{1}{8\pi}\int_\Sigma |\bf H|.\] 
The positivity is proved by Shi-Tam \cite{Shi-Tam}  and Liu-Yau \cite{Liu-Yau2}. However, for a surface in the Minkowski spacetime, the above expression may not be zero \cite{OST}. The optimal isometric embedding equation gives the necessary correction, so the definition of Wang-Yau is {\it positive} in general, and zero for surfaces in the Minkowski spacetime, see \cite{Wang-Yau2}. In general, the optimal isometric embedding equation is difficult to solve. However, in a perturbative configuration, when a family of surfaces limit to a surface in the Minkowski spacetime, the optimal isometric embedding equation is solvable, subject to the positivity of the limiting mass \cite{Chen-Wang-Yau2}.

 \section{Large sphere limit at null infinity}
 
  In \cite{Chen-Wang-Yau1}, we evaluate the large sphere limit of quasilocal mass at $\mathscr{I}^+$ which recovers the Bondi-Sachs energy momentum. At a retarded time $u=u_0$, we consider the family of large spheres $\Sigma_r$ parametrized by the Bondi-Sachs coordinate $r$. The positivity of the Bondi mass guarantees the unique solvability of the optimal isometric embedding system \eqref{oiee} with a solution $(X_r, T_r)$. Suppose $X_r$ and  $T_r$ admit 
expansions:
\[ T_r=T^{(0)}+
\sum_{k=1}^\infty T^{(-k)}r^{-k}\]
\[X_r=rX^{(1)}+X^{(0)}+\sum_{k=1}^\infty X^{(-k)}r^{-k},\] then 
$T^{(0)}$ is shown  to be proportional to the Bondi-Sachs energy-momentum  in \cite[Theorem 2]{Chen-Wang-Yau1}  and $T^{(0)}$ being future timelike makes $T^{(-k)}$ and $X^{(-k+1)}$ solvable inductively for $k=1, 2\cdots$.

In \cite{Chen-Wang-Yau_hyperbolic}, we also evaluate the large sphere limit of quasilocal mass on an asymptotically hyperbolic initial data set. This in particular gives a new proof of the positivity of the Bondi mass.

\section{Unit sphere limit at null infinity}
 
Both the positivity of Bondi mass \eqref{pmt} and the mass loss formula \eqref{energy_loss1} are global statements 
about $\mathscr{I}^+$ that require the information in all direction of $(\theta, \phi)$ on $S^2_\infty$. In \cite{Chen-Wang-Yau6, Vaidya, CWWY1, CWWY2, CWWY3}, we study the limit of quasilocal mass along a single direction $(\theta, \phi)$ and obtain quasilocal quantities
at $\mathscr{I}^+$.  Consider a null geodesic $\gamma$ with affine parameter $s$ that approaches $\mathscr{I}^+$. Around each point $\gamma(s)$, consider a geodesic 2-sphere $\Sigma_s$ of unit radius. 
 The geometry of $\Sigma_s$ approaches the geometry of a standard unit round sphere of $\R^3$.

 \begin{center}
{\includegraphics[height=5cm]{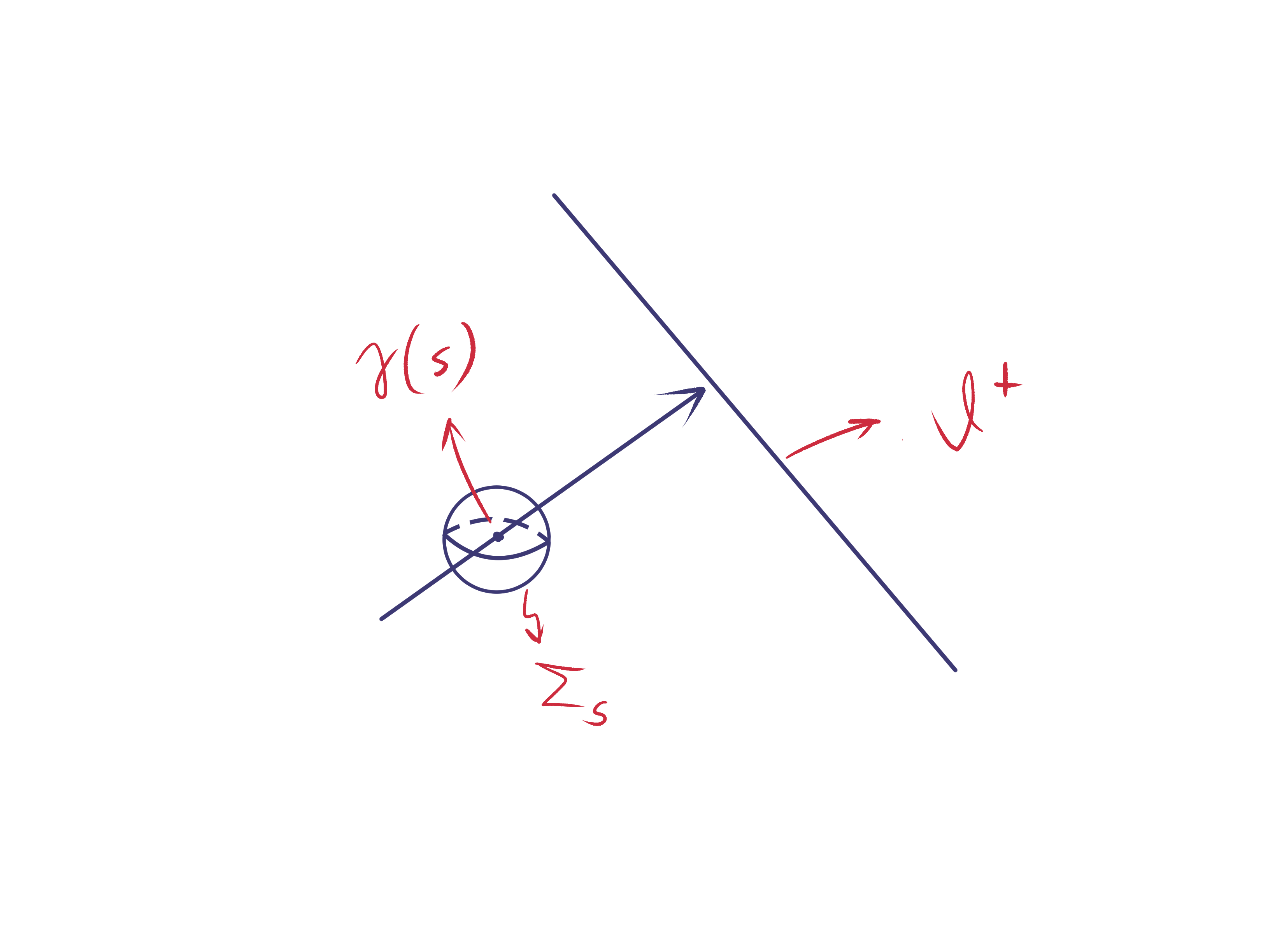}}
\end{center}

  In the limit $s\rightarrow \infty$, we obtain the limit of the quasilocal mass $\lim_{s\rightarrow \infty} E(\Sigma_s)$ which is of the order of $\frac{1}{s^2}$ with $E(\Sigma_s)\geq 0$. In \cite{Chen-Wang-Yau6}, we study the case of  linear gravitational perturbation of the Schwarzschild black hole \`a la Chandrasekhar \cite{Chandrasekhar}. The linearized vacuum Einstein equation is solved by separation of variables
and solutions of linearized waves are obtained. The optimal isometric embedding system can be solved and the quasilocal mass can be evaluated by solving equations of the following forms on the standard 2-sphere $S^2$:
\begin{equation}\label{ode} \begin{split}\Delta (\Delta+2)\tau&=\text{physical  data},\\
(\Delta+2)N&=\text{physical  data},\end{split}\end{equation} where $\tau$ and $N$ are functions on the standard 2-sphere and $\Delta$ is the Laplace operator. All distinctive features of the linearized waves such as frequency and mode parameters are recovered.

 In \cite{Vaidya}, we study the case of the Vaidya spacetime. The metric of the Vaidya spacetime takes the form:
\[-(1-\frac{2m(u)}{r})du^2-2dudr+r^2 (d\theta^2+\sin^2\theta d\phi^2).\] The quasilocal mass of a unit sphere approaching null infinity is computed in \cite{Vaidya}:
\[E(\Sigma_s)=-\frac{1}{16\pi s^2} \int_{S^2} (\partial_u m) f^2+l.o.t. \geq 0\]
In particular, the positivity of quasilocal mass corresponds to the mass loss formula in the Vaidya case.

 What happens in the general case? One may expect that the limit of quasilocal mass in the direction of $(\theta_0, \phi_0)$ should recover the mass aspect function $m(\cdot, \theta_0, \phi_0)$. But notice that 
the mass aspect function is not pointwise positive, only the integrated Bondi mass is positive. 
Besides, the positivity of Bondi mass requires a global assumption on horizon \cite{SY, HP}. For example, on a Schwarzschild spacetime with negative mass parameter (thus there is a naked singularity that is not shielded by a horizon),
the Bondi mass is negative, but the quasilocal mass is still positive. On the other hand, the Vaidya spacetime is non-vacuum, indeed both the dominant energy condition and the mass loss formula correspond to $\partial_u m \leq 0$.

In order to understand such a quasilocal mass at the purely gravitational level, we compute the case for the null infinity of a general spacetime in a Bondi-Sachs coordinate system \cite{CWWY2, CWWY3}. 
The $\frac{1}{s^2}$ term of the quasilocal mass of a unit sphere approaching null infinity is (up to a constant factor)
\begin{equation}\label{qlr} \int_{B^3} [\frac{1}{8} |\partial_u C|^2+\det (h_0^{(-1)}-h^{(-1)})]+\frac{1}{4}\int_{S^2}[ (tr_\Sigma k^{(-1)})^2-\tau^{(-1)}{\Delta}({\Delta}+2)\tau^{(-1)}],\end{equation} in which
$h^{(-1)}$ and $k^{(-1)}$ are part of the physical data and $h_0^{(-1)}$ and $\tau^{(-1)}$ depend on the solution of the optimal isometric embedding system. A priori, the expression may depend on the mass aspect and the shear tensor.  However, some rather amazing cancellations show that the answer depends only  on the shear tensor:

\begin{theorem}\cite{CWWY2, CWWY3}
 The limit of the quasilocal mass of unit-size sphere at $\mathscr{I}^+$ is a positive quasilocal quantity that depends only on the shear tensor.  
 \end{theorem}
This quantity should be considered as a quasilocal measurement of radiation. It is very interesting to compare \eqref{qlr} with the global loss formula \eqref{general_loss}. The first term $\frac{1}{8} \int_{B^3} |\partial_u C|^2$ in \eqref{qlr} should be considered as the principal term of radiation, which also appears in \eqref{general_loss}, and other terms in \eqref{qlr} should be considered as correction terms due to the quasilocal nature of the quantity.

%Given a spacetime surface $(\Sigma, \sigma, \bf{H})$, find $(\Sigma_0, \sigma, \bf{H}_0)$ through the optimal isometric embedding
%equation and evaluate the quasilocal mass. 
%This is a nonlinear and coordinate independent theory. The procedure is canonical and does not involve any ad hoc referencing or normalization.

%\end{frame}


\begin{thebibliography}{Besse}

%\bibitem{Bartnik1} R. Bartnik, 
%\textit{Quasi-spherical metrics and prescribed scalar curvature}, 
%J. Differential Geom. 37 (1993), no. 1, 31--71.
%\bibitem{Bartnik2} R. Bartnik, \textit{New definition of quasi-local mass}, Phys. Rev. Lett. {\bf 62} (1989), no.~20, 2346--2348.

\bibitem{Bondi} H. Bondi, \textit{Gravitational waves in general relativity}, Nature, 186:535, May 1960.
\bibitem{BVM}  H. Bondi,  M. G. J. van der Burg, and  A. W. K. Metzner, \textit{ Gravitational waves in general relativity. VII. Waves from axi-symmetric isolated systems}, Proc. Roy. Soc. Ser. A 269 (1962) 21--52.
\bibitem{Booth-Mann} I. S. Booth and R. B. Mann, Phys. Rev. D 59, 064021
(1999).


\bibitem{BY2} J. D. Brown and J. W.  York,  \textit{Quasi-local energy and conserved charges derived from the gravitational action,} Phys. Rev. D (3) \textbf{47} (1993), no. 4, 1407--1419.

\bibitem{Chandrasekhar} S. Chandrasekhar, {\it The mathematical theory of black holes}, reprint of the 1992 edition, Oxford Classic Texts in the Physical Sciences, Oxford Univ. Press, New York.


\bibitem{Chen-Wang-Yau1} P.-N. Chen, M.-T. Wang, and S.-T. Yau, \textit {Evaluating quasi-local energy and solving optimal embedding equation at null infinity}, Comm. Math. Phys. \textbf{308} (2011), no.3, 845--863.
\bibitem{Chen-Wang-Yau2}  P.-N. Chen, M.-T. Wang, and S.-T. Yau, \textit{Minimizing properties of critical points of quasi-local energy}, Comm. Math. Phys. \textbf{329} (2014), no.3, 919--935
\bibitem{Chen-Wang-Yau3}  P.-N. Chen, M.-T. Wang, and S.-T. Yau, \textit{Conserved quantities in general relativity: from the quasi-local level to spatial infinity}, Comm. Math. Phys. \textbf{338} (2015), no.1, 31--80.
\bibitem{Chen-Wang-Yau4}  P.-N. Chen, M.-T. Wang, and S.-T. Yau, \textit{Quasilocal angular momentum and center of mass in general relativity},  Adv. Theor. Math. Phys. 20, no. 4 (2016), 671--682.
\bibitem{Chen-Wang-Yau6}  P.-N. Chen, M.-T. Wang,\ and\ S.-T. Yau, \textit{Quasi-local energy in presence of gravitational radiation}, Int. J. Mod. Phys. D 25, 164501 (2016).
%\bibitem{Chen-Wang-Yau4}  P.-N. Chen, M.-T. Wang, and\ S.-T. Yau, \textit {Quasi-local mass in the gravitational perturbations of black holes,} in preparation. 

\bibitem{Chen-Wang-Yau_hyperbolic} P.-N. Chen, M.-T. Wang, \ and\ S.-T. Yau \textit{Conserved quantities on asymptotically
hyperbolic initial data sets}, Adv. Theor. Math. Phys. 20 (2016), no. 6, 1337--1375. arXiv: 1409.1812

\bibitem{Chen-Wang-Yau5} P.-N. Chen, M.-T. Wang, and \ S.-T. Yau, \textit{Evaluating small sphere limit of the Wang-Yau quasi-local energy}, Comm. Math. Phys. 357 (2018), no. 2, 731--774


\bibitem{Vaidya} P.-N. Chen, M.-T. Wang, and\ S.-T. Yau, {\it Quasi-local mass at the null infinity of the Vaidya spacetime}, Nonlinear analysis in geometry and applied mathematics, 33--48, Harv. Univ. Cent. Math. Sci. Appl. Ser. Math., 1, Int. Press, Somerville, MA, 2017,  arXiv:1608.06165

\bibitem{CWWY1} P.-N. Chen, Y.-K. Wang, M.-T. Wang, and S.-T. Yau \textit{Quasi-local mass on unit spheres at spatial infinity}, arXiv: 1901.06954
\bibitem{CWWY2} P.-N. Chen, Y.-K. Wang, M.-T. Wang, and S.-T. Yau \textit{Quasi-local mass at null infinity in Bondi-Sachs coordinates}, arXiv: 1901.06952
\bibitem{CWWY3} P-N. Chen, Y.-K. Wang, M.-T. Wang, and S.-T. Yau \textit{Quasi-local mass at axially symmetric null infinity}, arXiv: 1901.06948


\bibitem{CK} D. Christodoulou and S. Klainerman, \textit{The global nonlinear stability of the Minkowski space}, Princeton Mathematical Series, 41. Princeton University Press, Princeton, NJ, 1993.





\bibitem{CJK} Chru\'sciel, Piotr T.; Jezierski, Jacek; Kijowski, Jerzy, \textit{Hamiltonian field theory in the radiating regime,} Lecture Notes in Physics. Monographs, 70. Springer-Verlag, Berlin, 2002. 

\bibitem{CJS} Chru\'sciel, Piotr T.; Jezierski, Jacek; Szymon, Leski, \textit{The Trautman-Bondi mass of hyperboloidal initial data sets,}  Adv. Theor. Math. Phys. 8 (2004) 83--139. 

\bibitem{CMS} Chru\'ciel, Piotr T.; MacCallum, Malcolm A. H.; Singleton, David B. \textit{Gravitational waves in general relativity. XIV. Bondi expansions and the "polyhomogeneity'' of $\mathscr{I}$,} Philos. Trans. Roy. Soc. London Ser. A 350 (1995), no. 1692, 113--141.




%\bibitem{Dougan-Mason} A. J. Dougan and L. J. Mason, \textit{Quasilocal mass constructions with positive energy}, Phys.
%Rev. Lett., 67, 2119--2122, (1991).


%\bibitem{Christodoulou} D. Christodoulou, \textit{Nonlinear nature of gravitation and gravitational-wave experiments}, Phys. Rev. Lett. \textbf{67} (1991), no. 12, 1486--1489. 

%\bibitem{Hawking} S. W. Hawking, \textit{Gravitational radiation in an expanding universe},
%J. Math. Phys. \textbf{9}, 598 (1968).

%\bibitem{HH} S. W. Hawking and G. T.  Horowitz,
%\textit{The gravitational Hamiltonian, action, entropy and surface terms},
%Classical Quantum Gravity \textbf{13} (1996), no. 6, 1487--1498.

\bibitem{Geroch} R. Geroch, \textit{Asymptotic structure of space-time.} (Proc. Sympos., Univ. Cincinnati, Cincinnati, Ohio, 1976), pp. 1--105. Plenum, New York, 1977.

\bibitem{HPS} S. W. Hawking, M. J. Perry, and A. Strominger, \textit{Superrotation charge and supertranslation
hair on black holes.}, Journal of High Energy Physics, 2017(5):161, 2017.

\bibitem{HP}
G. T. Horowitz and M. J.  Perry,
\textit{Gravitational energy cannot become negative},
Phys. Rev. Lett. 48 (1982), no. 6, 371--374. 

\bibitem{HI} G.~Huisken and T. Ilmanen, \textit{The inverse mean curvature flow and the Riemannian Penrose inequality,} J. Diff. Geom. 59, 353--437 (2001).

\bibitem{HYZ} Wen-Ling Huang, Shing-Tung Yau, and Xiao Zhang,
\textit{Positivity of the Bondi mass in Bondi's radiating spacetimes}, 
Atti Accad. Naz. Lincei Rend. Lincei Mat. Appl. 17 (2006), no. 4, 335--349. 

\bibitem{KI} J. Kijowski,  \textit{A simple
derivation of canonical structure and quasi-local
 Hamiltonians in general relativity,} Gen. Relativity Gravitation \textbf{29} (1997), no. 3, 307--343.


\bibitem{Liu-Yau1}
C.C. M.~Liu and S.T.~Yau, \textit{Positivity of quasilocal mass,} Phys. Rev. Lett. 90, 231102 (2003)

\bibitem{Liu-Yau2} C.-C. M. Liu and S.-T.  Yau,
\textit{Positivity of quasi-local mass II,} J. Amer. Math. Soc. \textbf{19} (2006), no. 1, 181--204.



\bibitem{MW} T. M\"adler and J. Winicour, \textit{Bondi-Sachs formalism}, Scholarpedia, 11 (12): 33528, 2016. arXiv:1609.01731

\bibitem{OST}N. \'O Murchadha, L. B. Szabados\ and\ K. P. Tod, Comment on: ``Positivity of quasi-local mass''   Phys. Rev. Lett. {\bf 92} (2004), no.~25, 259001, 1 p. 



\bibitem{NP1} Newman, E. T.; Penrose, R. \textit{An approach to gravitational radiation by a method of spin coefficients,} J. Mathematical Phys. 3 1962, 566--578. 
\bibitem{NP2} Newman, E. T.; Penrose, R. \textit{Note on the Bondi-Metzner-Sachs group,} J. Mathematical Phys. 7 1966, 863--870. 


%\bibitem{Nirenberg} L. Nirenberg, \textit{The Weyl and Minkowski problems in differential geometry in
%the large}, Comm. Pure Appl. Math. 6 (1953), 337--394.



%\bibitem{Penrose1} R. Penrose,  \textit{Some unsolved problems in classical general relativity}, Seminar on Differential Geometry, pp. 631--668,  Ann. of Math. Stud., 102, Princeton Univ. Press, Princeton, N.J., 1982. 
%\bibitem{Penrose2} R. Penrose, \textit{Quasi-local mass and angular momentum in general relativity}, Proc. Roy. Soc. London Ser. A {\bf 381} (1982), no.~1780.

\bibitem{Penrose3} R. Penrose, \textit{Asymptotic properties of fields and space-times}, Phys. Rev. Lett. 10 1963 66--68. 

\bibitem{Penrose4} R.  Penrose, Republication of  \textit{Conformal treatment of infinity}, Gen. Relativity Gravitation 43 (2011), no. 3, 901--922. 

%\bibitem{Pogorelov} A. V. Pogorelov, \textit{Regularity of a convex surface with given Gaussian curvature,} (Russian) Mat. Sbornik N.S. \textbf{31}(73), (1952), 88--103.


%\bibitem{Rizzi} Rizzi, Anthony, \textit{Angular momentum in general relativity: a new definition,} Phys. Rev. Lett. 81 (1998), no. 6, 1150--1153.

\bibitem{Sachs} R. K. Sachs, {\it Gravitational waves in general relativity, VIII. Waves in asymptotically flat space-time.} Proc. Roy. Soc. Ser. A 270 1962 103--126.

\bibitem{Sauter} J. Sauter,  {\it Foliations of null hypersurfaces and the Penrose inequality.} PhD thesis, ETH Z\"urich,
2008.

\bibitem{SY} R. Schoen\ and\ S.-T. Yau, 
\textit{Proof that the Bondi mass is positive},
Phys. Rev. Lett. 48 (1982), no. 6, 369--371.

\bibitem{Shi-Tam}Y. Shi\ and\ L.-F. Tam, \textit{Positive mass theorem and the boundary behaviors of compact manifolds with nonnegative scalar curvature}, J. Differential Geom. {\bf 62} (2002), no.~1, 79--125.
%\bibitem{SY1} R. Schoen and S.-T. Yau, {\em On the proof of the positive mass conjecture in general relativity. } Comm. Math. Phys. {\bf 65} (1979), no. 1, 45--76. 
%\bibitem{SY2} R. Schoen and S.-T. Yau, {\em Proof of the positive mass theorem. II}, Comm. Math. Phys. {\bf 79} (1981), no. 2, 231--260. 

%\bibitem{Tod} K. P. Tod, Penrose's quasi-local mass, in {\it Twistors in mathematics and physics}, 164--188, London Math. Soc. Lecture Note Ser., 156, Cambridge Univ. Press, Cambridge.

\bibitem{Trautman1} A. Trautman, \textit{Boundary conditions at infinity for physical
theories}, Bull. Acad. Polon. Sci. 6 (1958), 403--406; reprinted
as arXiv:1604.03144. 

\bibitem{Trautman2} A. Trautman, \textit{Radiation and boundary conditions in the theory
of gravitation}, Bull. Acad. Polon. Sci., 6 (1958), 407--412;
reprinted as arXiv:1604.03145. 

\bibitem{VK} J. A. Valiente Kroon, \textit{A comment on the outgoing radiation condition for the gravitational field and the peeling theorem},  Gen. Relativity Gravitation 31 (1999), no. 8, 1219--1224. 

\bibitem{VDB} M. G. J. van der Burg \textit{ Gravitational waves in general relativity, IX. Conserved Quantities.} Proc. Roy. Soc. Ser. A 294 1966 112--122.

\bibitem{Wang-Yau1} M.-T. Wang, and S.-T. Yau, \textit{ Quasi-local mass in general relativity}, Phys. Rev. Lett. {\bf 102} (2009), no. 2, no. 021101.
\bibitem{Wang-Yau2}  M.-T. Wang, and S.-T. Yau, \textit{ Isometric embeddings into the Minkowski space and new quasi-local mass}, Comm. Math. Phys. {\bf 288} (2009), no. 3, 919--942. 
\bibitem{Winicour} J. Winicour, \textit{Logarithmic asymptotic flatness,} Found. Phys. 15 (1985), no. 5, 605--616.
%\bibitem{W} E. Witten,  {\em A new proof of the positive energy theorem.} %Comm. Math. Phys. 80 (1981), no. 3, 381--402. 

\end{thebibliography}
\end{document}